\documentclass[3p, twocolumn]{elsarticle}

\usepackage{amsthm} \usepackage{amssymb} \usepackage{amsmath} \usepackage{algorithm} \usepackage{algorithmic} \usepackage{graphicx}

\newtheorem{definition}{Definition} \newtheorem{lemma}{Lemma} \newtheorem{theorem}{Theorem}  \newtheorem{corollary}{Corollary} \newtheorem{Remark}{Remark}
\newtheorem{property}{Property}

\begin{document}

\begin{frontmatter}

\title{A Comprehensive Study of an Online Packet Scheduling Algorithm}

\author{Fei Li}

\address{Department of Computer Science\\
George Mason University\\
Fairfax, Virginia 22030\\
Email: \textsf{lifei@cs.gmu.edu}}


\begin{abstract}
We study the \emph{bounded-delay model} for Qualify-of-Service buffer management. Time is discrete. There is a buffer. Unit-length jobs (also called \emph{packets}) arrive at the buffer over time. Each packet has an integer release time, an integer deadline, and a positive real value. A packet's characteristics are not known to an online algorithm until the packet actually arrives. In each time step, at most one packet can be sent out of the buffer. The objective is to maximize the total value of the packets sent by their respective deadlines in an online manner. An online algorithm's performance is usually measured in terms of \emph{competitive ratio}, when this online algorithm is compared with a clairvoyant algorithm achieving the best total value. In this paper, we study a simple and intuitive online algorithm. We analyze its performance in terms of competitive ratio for the general model and a few important variants.
\end{abstract}


\end{frontmatter}


\section{Model Description}

We consider the \emph{bounded-delay model} introduced in~\cite{Hajek01, KesselmanLMPSS04}. Time is discrete. The $t$-th \emph{(time) step} presents the time interval $(t - 1, \ t]$. There is a buffer and unit-length jobs (also called \emph{packets}) arrive at the buffer over time. Each packet $p$ has an integer release time $r_p \in \mathbb Z^+$, an integer deadline $d_p \in \mathbb Z^+$, and a positive real value $v_p \in \mathbb R^+$. A packet $p$'s characteristics are not known to an online algorithm until $p$ actually arrives at the buffer at time $r_p$. In each step, at most one packet in the buffer can be sent. A packet $p$ is said to be \emph{successfully sent} at time $t$ if $r_p \le t \le d_p$. The objective is to maximize the total value of the packets that are successfully sent in an online manner.

As people have noted, the offline version of this problem can be solved efficiently using the Hungarian algorithm~\cite{Kuhn55} in time $O(n^3)$, where $n$ is the number of packets in the input instance.

In the framework of \emph{competitive analysis} which provides worst-case guarantees, an online algorithm's performance is measured in terms of \emph{competitive ratio}~\cite{BorodinE98}. For a maximization problem, an online algorithm is called \emph{$c$-competitive} if for \emph{any} finite instance, its total value is no less than $1 / c$ times of what an optimal offline algorithm achieves. In competitive analysis, an input instance is allowed to be generated in an adversarial way so as to maximize the competitive ratio. The upper bound of competitive ratio is achieved by some online algorithms. A competitive ratio strictly less than the lower bound cannot be reached by any online algorithm. If an online algorithm has its competitive ratio same as the lower bound, we say that this online algorithm is \emph{optimal}. For the bounded-delay model, the currently best known result is $2 \sqrt{2} - 1 \approx 1.828$~\cite{EnglertW07} and the lower bound is $(1 + \sqrt{5}) / 2 \approx 1.618$~\cite{Hajek01, ChinF03}. If an online algorithm decides which packet to send only based on the contents of its current buffer, and independent of the packets that have already been released and processed, we call it \emph{memoryless}.

In this paper, we study a simple, intuitive memoryless online algorithm called MG (`Modified Greedy'). We analyze MG's performance in terms of competitive ratio for the general bounded-delay model and some important variants. Define a packet $p$'s \emph{slack-time} $s_p$ as the difference between its deadline $d_p$ and release time $r_p$, $s_p = d_p - r_p$. The variants that we consider include:
\begin{itemize}
\item \emph{Agreeable deadline} setting. In an agreeable deadline instance, for any two packets $p$ and $q$ with $r_p \le r_q$, we have $d_p \le d_q$. This variant has been studied in~\cite{LiSS05}.

\item \emph{Anti-agreeable deadline} setting. In an anti-agreeable deadline instance, for any two packets $p$ and $q$ with $r_p \le r_q$, we have $d_p \ge d_q$.

\item \emph{Agreeable value} setting. In an agreeable value instance, for any two packets $p$ and $q$ with $r_p \le r_q$, we have $v_p \le v_q$.

\item \emph{Anti-agreeable value} setting. In an anti-agreeable value instance, for any two packets $p$ and $q$ with $r_p \le r_q$, we have $v_p \ge v_q$.

\item \emph{Agreeable deadline/value} setting. In an agreeable deadline/value instance, for any two packets $p$ and $q$ with $d_p \le d_q$, we have $v_p \le v_q$.

\item \emph{Anti-agreeable deadline/value} setting. In an anti-agreeable deadline/value instance, for any two packets $p$ and $q$ with $d_p \le d_q$, we have $v_p \ge v_q$.

\item \emph{Agreeable slack-time/value} setting. In an agreeable slack-time/value instance, for any two packets $p$ and $q$ with $s_p \le s_q$, we have $v_p \le v_q$.

\item \emph{Anti-agreeable slack-time/value} setting. In an anti-agreeable slack-time/value instance, for any two packets $p$ and $q$ with $s_p \le s_q$, we have $v_p \ge v_q$.
\end{itemize}

Our results are summarized in Table~\ref{tbl:result}. Note that the lower bounds shown in Table~\ref{tbl:result} are the lower bounds of MG's performance but not the lower bounds for any online algorithms.

\begin{table*}
\centering
\begin{tabular}{|c|c|c|p{3cm}|}
\hline
models & upper bounds & lower bounds & notes  \\ \hline \hline

general & $2$ & $2$~\cite{LiSS07} & A detailed analysis of the lower bound is given in this paper. \\ \hline

agreeable deadline & $\phi$~\cite{LiSS05} & $\phi$~\cite{ChinF03} & MG is optimal. \\ \hline
anti-agreeable deadline & $2$ & $2$~\cite{LiSS07} & - \\ \hline

agreeable value & $2$ & $2$~\cite{LiSS07} & - \\ \hline
anti-agreeable value & $1$ & $1$ & MG is optimal. \\ \hline

agreeable deadline/value & $\phi$ & $\phi$~\cite{ChinF03} & MG is optimal. \\ \hline
anti-agreeable deadline/value & $1$ & $1$ & MG is optimal. \\ \hline

agreeable slack-time/value & $\phi$ & $1$ & - \\ \hline
anti-agreeable slack-time/value & $1$ & $1$ & MG is optimal. \\ \hline
\end{tabular}
\caption{Summary of MG's performance for the bounded-delay model and its variants. The results without references are the work presented in this paper. In this table, $\phi = (1 + \sqrt{5}) / 2 \approx 1.618$.}
\label{tbl:result}
\end{table*}

In the following, we present the online algorithm MG in Section~\ref{sec:mg} and analyze its performance in Section~\ref{sec:analysis}.


\section{Algorithm MG}
\label{sec:mg}

The idea of designing MG is motivated by the \emph{greedy algorithm}: In each step, the highest-value pending packet is sent. This algorithm is proved $2$-competitive~\cite{Hajek01, KesselmanLMPSS04}. In one attempt to beat the greedy algorithm in competitiveness, Chin et al.~\cite{ChinCFJST06} proposed an algorithm called EDF$_\alpha$, bearing the idea of sending the earliest-deadline packet with a sufficiently large value (for instance, at least $1 / \alpha$ times of the highest value of a pending packet where $\alpha \ge 1$). Note that EDF$_\alpha$ generalizes the greedy algorithm, which is EDF$_1$. Same as the greedy algorithm, EDF$_\alpha$ is asymptotically not better than $2$-competitive. For EDF$_\alpha$, it is possible that the expiring packet in the algorithm's buffer is the one that an optimal offline algorithm sends and this packet has only a slightly less value than the packet that EDF$_\alpha$ sends.

Recall that a memoryless online algorithm makes its decision only based on the contents of its current buffer. Thus, it is natural to send a packet from a set of packets, all of which are eligible of being sent successfully under the assumption of no future arrivals. We consider \emph{provisional schedules}. A provisional schedule~\cite{ChrobakJST07a, EnglertW07} at time $t$ is a schedule specifying the set of pending packets to be transmitted and for each it specifies the delivery time, assuming no newly arriving packets. An optimal provisional schedule achieves the maximum total value among all the provisional schedules. At the beginning of each step, we calculate an optimal provisional schedule $S$ and the packets in $S$ are arranged in a \emph{canonical} order: increasing order of deadlines, with ties broken in decreasing order of values.

Let $e$ denote the first packet in $S$ and $h$ denote the first highest-value packet in $S$. Motivated by the idea of EDF$_\alpha$, we would like to send a packet with a sufficiently large value compared with $v_h$. At the same time, from the tight example for EDF$_\alpha$, we would like to send a packet to compensate the potential loss due to not sending the earliest-deadline packet $e$. Thus, we send a packet $f$ in the optimal provisional schedule satisfying $v_f \ge v_h / \alpha$ if $f = e$ and $v_f \ge \max\{\beta v_e, \ v_h / \alpha\}$ if $f \neq e$, where $\alpha, \ \beta \ge 1$. In order to guarantee that at least one packet in $S$ can be a candidate packet for $f$, we have to have $\alpha \ge \beta$ since if $v_e < v_h / \alpha$, we should have $v_h \ge v_f \ge \max\{\beta v_e, \ \alpha v_e\} \ge \max\{\beta, \ \alpha\} v_e$. The algorithm MG is described in Algorithm~\ref{alg:mg}.

\begin{algorithm}
\caption{\textsc{MG} ($t, \ 1 \le \beta \le \alpha$)}
\begin{algorithmic}[1]
\STATE Calculate an optimal provisional schedule $S$. All the packets in $S$ are sorted in a canonical order: increasing order of deadlines, with ties broken in decreasing order of values. In $S$, let $e$ denote the first packet; let $h$ denote the first highest-value packet.

\IF{$v_e \ge v_h / \alpha$}

\STATE send $e$;

\ELSE

\STATE send the first packet $f$ satisfying $v_f \ge \max\{v_h / \alpha, \ \beta v_e\}$.

\ENDIF
\end{algorithmic}
\label{alg:mg}
\end{algorithm}

Note that MG generalizes EDF$_\alpha$ (and the greedy algorithm). If $\alpha = 1$ (hence $\beta = 1$ since $\alpha \ge \beta \ge 1$), MG is the greedy algorithm.  If $\beta = 1$, MG is no-worse than EDF$_\alpha$ in competitiveness.

\begin{theorem}
If $\beta = 1$, MG is no-worse than EDF$_\alpha$ in competitiveness.
\end{theorem}

\begin{proof}
We inductively prove that ($1$) MG with $\beta = 1$ and EDF$_\alpha$ share the same buffer at any time; (However, we note here that MG's optimal provisional schedule may not be identical to EDF$_\alpha$'s buffer.) and ($2$) in each step, the charged value to MG is no less than the charged value to EDF$_\alpha$.

Assume MG sends $f \neq e$. EDF$_\alpha$ must send $f$ as well since all the packets with values $\ge v_h / \alpha$ must be in MG's optimal provisional schedule. Assume MG sends the $e$-packet and EDF$_\alpha$ sends a packet $p$ not in MG's optimal provisional schedule. If EDF$_\alpha$ does not send $e$ in its schedule, we have $v_e \ge v_p$ and we can use $e$ to replace $p$ for EDF$_\alpha$.
\end{proof}


\section{Analysis}
\label{sec:analysis}

Let OPT denote an optimal offline algorithm and $\mathcal O$ denote the set of packets that OPT sends. Let ADV denote a (modified) adversary. In our proof, we will create ADV and make sure that ADV gains a total value no less than $\sum_{p \in {\mathcal O}} v_p$.


\subsection{The general setting}

\begin{theorem}
MG is $2$-competitive for the bounded-delay model, for any $1 \le \beta \le \alpha \le 2$.
\label{theorem:ub}
\end{theorem}

\begin{proof}
We assume that there exists an adversary called ADV. We modify ADV such that ADV and MG share the same buffer at the beginning of each step. ADV does not have to send every packet in its buffer. In a step, MG sends the packet $f$.

\begin{enumerate}
\item Assume ADV sends the same packet $f$ in this step.

ADV and MG gain the same value.

\item Assume ADV sends a packet $j$ ($\neq f$) with $d_j < d_f$.

We modify ADV by sending both $j$ and $f$ in the current step. We then insert $j$ into ADV's buffer as a gift packet. As assumed, $j$ is in MG's buffer at the beginning of this step. From the canonical order and MG choosing $f$ but not $j$ to send, we have $v_j \le v_f$. Then $v_j + v_f \le 2 v_f$.

\item Assume ADV sends a packet $j$ ($\neq f$) with $d_j > d_f$.

As assumed, $j$ is in MG's buffer at the beginning of this step. No matter $f = e$ or $f \neq e$, we have $v_f \ge v_h / \alpha \ge v_j / \alpha \ge v_j / 2$. Note that $v_f < v_j$ (and $d_f < d_j$) since otherwise, ADV prefers to sending $f$ instead of $j$. We then insert $j$ into ADV's buffer to replace $f$.
\end{enumerate}

At the end of this step, ADV and MG share the same buffer again. The modifications that we make favor the adversary but not MG. In this step, ADV's modified gain is bounded by $2$ times of what MG achieves.
\end{proof}

\begin{theorem}
MG is asymptotically no better than $2$-competitive for the bounded-delay model, with $\alpha = \beta = \phi$.
\label{theorem:lb}
\end{theorem}

A sketched proof of Theorem~\ref{theorem:lb} has been given in a conference paper~\cite{LiSS07}. We detail the analysis in journal paper.

\begin{proof}
We construct an example to prove Theorem~\ref{theorem:lb}. We use $\infty$ in the deadline field of a packet to show that this packet's deadline is very large. Let $n = 2^k$.  The packets are released in a stage-manner. There are $\log n = k$ stages. The superscript of a packet shows the stage in which it is released.

At the beginning of step $1$, there are $3$ packets in MG's buffer. The adversary has the same buffer. These $3$ packets are $e^1_1 := (1 + \epsilon, \ 2)$, $f^1_1 := (\phi - \epsilon, \ 2^{k + 1} - k)$, and $h^1_1 := (\phi, \ \infty)$. MG sends $h^1_1$, and $e^1_1$ is dropped out of the buffer due to its deadline.

In each of the following ($2^k - k + 1$) time steps, say step $i$, a group of $3$ packets are released: $e^1_i := (1 + \epsilon, \ i + 1)$, $f^1_i := (\phi - \epsilon, \ 2^{k + 1} - k)$, and $h^1_i := (\phi, \ \infty)$. In step $i$, MG sends $h^1_i$ and drops $e^1_i$ due to its deadline. At the end of the ($2^k - k + 1$)-th step, MG's buffer is full of ($2^k - k + 1$) $f^1_i$-packets ($\forall i = 1, \ 2, \ \ldots, \ 2^k - k + 1$). The first stage ends. The length of stage $1$ guarantees that no $f^1_i$ packet, especially packet $f^1_1$, becomes the first packet in the buffer.

At the beginning of step $2^k - k + 1$, the second stage starts. The adversary releases a pair of packets $f^2_1 := (\phi (\phi - \epsilon) - \epsilon, \ 2^{k + 1} - k + 1)$ and $h^2_1 := (\phi^2, \ \infty)$. The newly released packets have later deadlines and are sorted canonically after the packets already in MG's buffer. MG sends $h^2_i$. Stage $2$ contains $2^{k - 1} - k + 2$ steps. The length of stage $2$ guarantees that no packet $f^2_i$ becomes the first packet in the buffer. In each of those $2^{k - 1} - k + 2$ steps, say step $i$, $2$ packets are released $f^2_i := (\phi (\phi - \epsilon) - \epsilon, \ 2^{k + 1} - k + 1)$ and $h^2_i := (\phi^2, \ \infty)$. MG sends $h^2_i$ in step $i$. Stage $2$ is half as long as stage $1$.

We repeat this pattern in each stage, for $k$ stages.  Stage $i + 1$ is half as long as stage $i$.  In each step $j$ of stage $i$, $2$ packets are released, $f^i_j := (\phi (w_{f^{i - 1}_1} - \epsilon), \ 2^{k + 1} - k + i)$ and
$h^i_j := (\phi^i, \ \infty)$. MG sends $h^i_j$ in step $j$. In the last stage, which is step $2^{k + 1}$, the adversary only releases $2$ packets $f^k_1 := (\phi^k, \ 2 n)$ and $h^k_1 := (\phi^{k + 1} + \epsilon, \ \infty)$. MG sends $h^k_1$ and $f^k_1$ is dropped out of the buffer due to its deadline.

For each step in stage $i$, MG only delivers the $h^i$ packets, and eventually, all packet $f^i$ are dropped out of the buffer due to their deadlines. On the contrary, the adversary sends all $f^i$ packets and all $h^i$ packets. A routine calculation shows that the optimal weighted throughput is nearly twice MG's weighted throughput. We remove $\epsilon$ in the following calculation for the sake of clearness.

\begin{eqnarray*}
c & = & \frac{2 \left(\phi^0 \cdot 2^k + \phi^1 \cdot 2^{k - 1} + \ldots + \phi^k \cdot 2^0\right) + \phi^{k + 1}}{\left(\phi^0 \cdot 2^k + \phi^1 \cdot 2^{k - 1} + \ldots + \phi^k \cdot 2^0\right) + \phi^{k + 1}}\\
& = & \frac{2 \left(\phi^0 \cdot 2^k\right) \left(\frac{\phi^0}{2^0} + \frac{\phi^1}{2^1} + \frac{\phi^2}{2^2} + \ldots + \frac{\phi^k}{2^k}\right) + \phi^{k + 1}}{(\phi^0 \cdot 2^k) \left(\frac{\phi^0}{2^0} + \frac{\phi^1}{2^1} + \frac{\phi^2}{2^2} + \ldots + \frac{\phi^k}{2^k}\right) + \phi^{k + 1}}\\
& = & \frac{2^{k + 1} \frac{1 - \left(\frac{\phi}{2}\right)^{k + 1}}{1 - \frac{\phi}{2}} + \phi^{k + 1}}{2^k \frac{1 - \left(\frac{\phi}{2}\right)^{k + 1}}{1 - \frac{\phi}{2}} + \phi^{k + 1}}\\
& = & \frac{2^{k + 1} - \phi^{k + 1} + \phi^{k + 1} - \frac{\phi^{k + 2}}{2}}{2^k - \frac{\phi^{k + 1}}{2} + \phi^{k + 1} - \frac{\phi^{k + 2}}{2}}\\
& = & \frac{2 \left(\frac{2}{\phi}\right)^k - \frac{\phi^2}{2}}{\left(\frac{2}{\phi}\right)^k - \frac{1}{2}}\\
& = & 2.
\end{eqnarray*}
\end{proof}


\subsection{The agreeable deadline setting}

In~\cite{LiSS05}, the authors have shown that MG is $\phi$-competitive for agreeable deadline instances. The lower bound $\phi$ constructed in~\cite{ChinF03} for the general model holds as well for scheduling packets with agreeable deadlines and MG. We list MG's performance in the agreeable deadline setting here for its optimality and significance. We include this variant for comparison with others.


\subsection{The anti-agreeable deadline setting}

Both Theorem~\ref{theorem:ub} and Theorem~\ref{theorem:lb} hold for anti-agreeable deadline instances. Both the upper bound and lower bound for MG are $2$.


\subsection{The agreeable value setting}

Both Theorem~\ref{theorem:ub} and Theorem~\ref{theorem:lb} hold for anti-agreeable deadline instances. Both the upper bound and lower bound for MG are $2$.


\subsection{The anti-agreeable value setting}

\begin{theorem}
MG is $1$-competitive for the anti-agreeable value setting when $\alpha = \infty$. MG is optimal.
\label{theorem:antivalue}
\end{theorem}

\begin{proof}
When $\alpha = \infty$, MG sends the earliest-deadline packet $e$ in the optimal provisional schedule in each step. To prove Theorem~\ref{theorem:antivalue}, we only need to inductively show that for each step, an optimal offline algorithm OPT sends $e$ in each step as well. In anti-agreeable value instances, any later released packet has a value $\le v_e$. If any later released packet belongs to $\mathcal O$, so does $e$. If no later released packet belongs to $\mathcal O$, OPT sends $e$ to maximize its total gain. Thus, OPT sends $e$ in each step.
\end{proof}


\subsection{The agreeable deadline/value setting}

The lower bound $\phi$ constructed in~\cite{ChinF03} for the general model holds as well for agreeable deadline/value instances.

\begin{theorem}
MG is $\phi$-competitive for the agreeable deadline/value setting when $\alpha = \beta = \phi^2 \approx 2.618$. MG is optimal.
\label{theorem:newmg}
\end{theorem}

\begin{proof}
We are using a charging scheme to prove Theorem~\ref{theorem:newmg}. Let OPT denote an optimal offline algorithm. Without loss of generality, we assume that OPT only accepts $\mathcal O$-packets and sends them in EDF manner. Let $Q^\text{OPT}$ denote OPT's buffer.

At time $t$, let the optimal provisional schedule be $S$ and we index the buffer slots as $t, \ t + 1, \ \ldots$. The packets in $S$ are sorted in increasing deadline order, with ties broken in decreasing value order and these packets are buffered in slots $t, \ t + 1, \ \ldots, \ t + |S| - 1$ consecutively. The packets not in $S$ are appended at the end of $S$. Let us study the optimal provisional schedule $S$ at first. The packets in $S$ thus are grouped into multiple ($\ge 1$) \emph{batches of packets} $G_1, \ G_2, \ \ldots$, in order of \emph{strictly increasing deadlines}. The packets in the same batch share the same deadline. (Note that $G_1$ is the first batch in $S$.) We have
\begin{Remark}
All the packets in the same batch share the same deadline. For any two batches $G_i$ and $G_j$ with indexes $i < j$, all the packets in $G_i$ have strictly earlier deadlines and strictly lower values than all the packets in $G_j$.
\label{remark:nodelay}
\end{Remark}

We will introduce a charging scheme and this charging scheme may use the following observations.

\begin{Remark}
In the agreeable deadline/value setting, if a packet $p$ is inserted into the optimal provisional schedule, then all the packets with value $> v_p$ are shifted into one buffer slot later since they have strictly larger deadlines. Also, for any two time steps, the relative order among the packets in both MG's optimal provisional schedules is not changed.
\label{remark:sameorder}
\end{Remark}

\begin{lemma}
In the agreeable deadline/value setting, if a packet $p$ is evicted out of MG's optimal provisional schedule at time $t$, then in each step from time $t$ till $p$'s deadline $d_p$, MG's optimal provisional schedules for these steps do not contain any packet with a value $< d_p$.
\label{lemma:nodelay}
\end{lemma}

\begin{proof}
If a packet $p$ is evicted out of MG's optimal provisional schedule at time $t$, then either $d_p < t$ or in each of the buffer slots $t, \ t + 1, \ \ldots, \ d_p$, MG's current optimal provisional schedule at time $t$ buffers one packet with value $> v_p$. From Remark~\ref{remark:nodelay} and the assumption of agreeable deadline/value, $d_p$ should not be larger than those of packets in the batch $G_1$.

\begin{itemize}
\item Assume MG sends the $e$-packet in a step before $d_p$.

Then for those packets arranged in the buffer slots belonging to batch $G_1$, they have their deadlines no smaller than $d_p$ and they are \emph{tight}, that is, they cannot be shifted into later buffer slots and provide buffer slots to accommodate less-value packets with no-later deadlines (see Remark~\ref{remark:sameorder}). For packets in batches $G_2, \ G_3, \ \ldots$, if any, they have strictly larger deadlines than $d_p$ and strictly larger values than $v_p$.

\item Assume MG sends a packet $f \neq e$ in a step before $d_p$.

All the unsent packets in the optimal provisional schedule can be shifted by at most one step to their later steps and the relative order among all these packets keep unchanged (see Remark~\ref{remark:nodelay} and Remark~\ref{remark:sameorder}). Any newly released packets with later deadlines have no smaller values. Any newly released packets with values $< v_p$ are rejected by MG's optimal provisional schedules since all the packets with deadlines $= d_p$ are tight. Thus, for the new optimal provisional schedule generated at the beginning of the next step, Lemma~\ref{lemma:nodelay} still holds.
\end{itemize}
\end{proof}

\begin{lemma}
Consider a \emph{chain of $k$ steps}. In the steps $1, \ 2, \ \ldots, \ k$ (these steps may not be continues), we charge OPT the values $v_{q_1}, \ v_{q_2}, \ \ldots, \ v_{q_k}$ and MG the values $v_{p_1}, \ v_{p_2}, \ \ldots, \ v_{p_k}$, respectively. If for all $i$ with $1 \le i \le k - 1$, we have $v_{q_i} \le \alpha \cdot v_{p_i}$, and if $v_{q_i} \le v_{p_{i + 1}}$ and $v_{q_k} \le v_{p_k}$, then $\sum^k_{i = 1} v_{q_i} \le \frac{1}{\alpha^k - 1} \left(\left(2 - \frac{1}{\alpha}\right) \alpha^k - \alpha\right) \sum^k_{i = 1} v_{p_i}$.
\label{lemma:chain}
\end{lemma}

\begin{proof}
\begin{eqnarray*}
& & \frac{\sum^k_{i = 1} v_{q_i}}{\sum^k_{i = 1} v_{p_i}}\\
& = & \frac{v_{q_1} + v_{q_2} + \cdots + v_{q_k}}{v_{p_1} + v_{p_2} + \cdots + v_{p_k}}\\
& \le & \frac{v_{q_1} + v_{q_2} + \cdots + v_{q_k}}{\frac{v_{q_1}}{\alpha} + \max\{v_{q_1}, \ \frac{v_{q_2}}{\alpha}\} + \cdots + v_{p_k}}\\
& \le & \frac{\frac{v_{q_2}}{\alpha} + v_{q_2} + \cdots + v_{q_k}}{\frac{v_{q_2}}{\alpha^2} + \frac{v_{q_2}}{\alpha} + \cdots + v_{p_k}} \le \cdots\\
& \le & \frac{\frac{v_{q_{k - 1}}}{\alpha^{k - 2}} + \cdots + \frac{v_{q_{k - 1}}}{\alpha} + v_{q_{k - 1}} + v_{q_k}}{\frac{v_{q_{k - 1}}}{\alpha^{k - 1}} + \cdots + \frac{v_{q_{k - 1}}}{\alpha^2} + v_{p_{k - 1}} + v_{p_k}}\\
& \le & \frac{\frac{v_{q_{k - 1}}}{\alpha^{k - 2}} + \cdots + \frac{v_{q_{k - 1}}}{\alpha} + v_{q_{k - 1}} + v_{q_k}}{\frac{v_{q_{k - 1}}}{\alpha^{k - 1}} + \cdots + \frac{v_{q_{k - 1}}}{\alpha^2} + \frac{v_{q_{k - 1}}}{\alpha} + \max\{v_{q_k}, \ v_{q_{k - 1}}\}}\\
& \le & \frac{\frac{v_{q_{k - 1}}}{\alpha^{k - 2}} + \cdots + \frac{v_{q_{k - 1}}}{\alpha} + v_{q_{k - 1}} + v_{q_{k - 1}}}{\frac{v_{q_{k - 1}}}{\alpha^{k - 1}} + \cdots + \frac{v_{q_{k - 1}}}{\alpha^2} + \frac{v_{q_{k - 1}}}{\alpha} + v_{q_{k - 1}}}\\
& = & \frac{\frac{1 - \alpha^{1 - k}}{1 - \alpha^{-1}} + 1}{\frac{1 - \alpha^{-k}}{1 - \alpha^{-1}}}\\
& = & \frac{(2 - \alpha^{-1}) \alpha^k - \alpha}{\alpha^k - 1}.
\end{eqnarray*}
\end{proof}

Note that when $\alpha \ge 1$, $\frac{1}{\alpha^k - 1} \left(\left(2 - \frac{1}{\alpha}\right) \alpha^k - \alpha\right) \le 2 - \frac{1}{\alpha}$. Also, note $\phi + \frac{1}{\phi^2} = 2$, we have
\begin{corollary}
Consider a \emph{chain of $k$ steps}. In the steps $1, \ 2, \ \ldots, \ k$ (these steps may not be continues), we charge OPT the values $v_{q_1}, \ v_{q_2}, \ \ldots, \ v_{q_k}$ and ON the values $v_{p_1}, \ v_{p_2}, \ \ldots, \ v_{p_k}$. If for all $i$ with $1 \le i \le k - 1$, we have $v_{q_i} \le \alpha \cdot v_{p_i}$, and if $v_{q_i} \le v_{p_{i + 1}}$, and $v_{q_k} \le v_{p_k}$, then we have $\sum^k_{i = 1} v_{q_i} \le \phi \sum^k_{i = 1} v_{p_i}$ when $\alpha = \phi^2$.
\label{coro:1}
\end{corollary}

We say that a chain of steps is \emph{open} if we have not charged the values to OPT and MG in these steps. Otherwise, we say that it is \emph{closed}.

\begin{definition}[Canonical Order]
Packets in MG's optimal provisional schedule are order in a \emph{canonical order}: in increasing order of deadlines, with ties broken in decreasing order of values.
\end{definition}

Our charging scheme guarantees the following three invariants:
\begin{enumerate}
\item[$I_1$.] In each step or in a closed chain of a group of steps, the total charged values to OPT are bounded by $\phi$ times of the total charged values to MG. Chains do not share steps.

\item[$I_2$.] For any packet $q$ in OPT's buffer, if $v_q$ has not been charged to OPT in our charging scheme, then $q$ must map uniquely to a packet $p$ in MG's optimal provisional schedule with $v_q \le v_p$ and $d_q \le d_p$. ($p$ may be the packet $q$ itself.)

In the canonical order, for any packet $j$ before $p$ in MG's optimal provisional schedule $S$, if $p$ is not in $S$, then we have $v_j \ge v_q$.

\item[$I_3$.] A packet $p$ in MG's optimal provisional schedule $S$ may correspond to at most one open chain and $v_p$ is no less than the value of the packet OPT sends in the last step of this open chain. If $p$ corresponds to an open chain and is mapped by a packet in OPT's buffer, $p$ is called \emph{overloaded}. If $p$ is overloaded, then any packet before $p$ in $S$ is overloaded as well.
\end{enumerate}

Note that Invariant $I_1$ results in Theorem~\ref{theorem:newmg} automatically.

The charging scheme is described below. We consider packet arrivals and packet deliveries separately.


\subsubsection*{Packet arrivals.}

For any packet $p$ evicted out of MG's optimal provisional schedule $S$ due to accepting a new arrival $p'$, we have $v_{p'} \ge v_p$ and $d_{p'} \ge d_p$ in the agreeable deadline/value setting. After dropping $p$, MG has at least one packet $q$ in $S$ such that $q$ is not mapped by a packet in OPT's buffer, due to Invariant $I_2$. In the canonical order of $S$, we pick up the first packet not in mapping and let it be $q$. $q$ should have a deadline $\ge v_p$ and thus, $v_q \ge v_p$, due to the assumption of agreeable deadline/value setting. Furthermore, any packet in MG's current optimal provisional schedule has a no-less value and no-earlier deadline than $p$. We transfer the open chain mapping to $p$, if any, to $q$. Hence for packet arrivals, all the invariants hold.


\subsubsection*{Packet deliveries.}

In each step, OPT sends the earliest-deadline packet $q$ in its buffer. MG sends either $e$ or $f \neq e$. Remember that we use $S$ denotes MG's optimal provisional schedule and the packets in $S$ are sorted in a canonical order.


\paragraph{Assume MG sends $e$ and OPT sends $q \notin S$ or OPT sends $q = e$}

From Invariant $I_2$, if $q$ has not been charged to OPT, then $v_q \le v_e$. Assume $q$ maps to $p$ in $S$. $v_q \le v_p \le v_e$. We charge OPT $v_q$ and the packets in the open chain mapping to $e$, if any. We close the open chain. The ratio of total charged values of this chain or this single step is bounded by $\phi$ (see Corollary~\ref{coro:1}).


\paragraph{Assume MG sends $e$ and OPT sends $q \in S$ with $q \neq e$}

Due to Invariants $I_2$ and $I_3$, there is no overloaded packets in MG's optimal provisional schedule. Otherwise, OPT sends a packet with an earlier deadline than $d_q$ and less-value than $v_e$ since it sends packets in the EDF order. We start a new open chain from this step mapping to $q$ in MG's optimal provisional schedule. Note that $q$ is not an overloaded packet yet since it is not mapped by any packet in OPT's buffer.


\paragraph{Assume MG sends $f \neq e$ and OPT sends $q \notin S$ or OPT sends $q = e$}

From Algorithm~\ref{alg:mg}, we have $v_f \ge \alpha v_e = \phi^2 v_e$. If $q$ is evicted out of the provisional schedule, we have $v_q \le v_e$ (from Lemma~\ref{lemma:nodelay}). We close the open chain if $e$ belongs to any one. The ratio of total charged values of this chain or this single step is bounded by $\phi$ (see Corollary~\ref{coro:1}).


\paragraph{Assume MG sends $f \neq e$ and OPT sends $q \in S$ with $d_q < d_f$}

\begin{itemize}
\item Assume $f = h$.

We have $v_q < v_h / \alpha = v_f / \alpha = v_f / \phi^2$.

If $q = e$, we close the open chain mapping to $e$, if any. We also charge $v_h$ to OPT in this step. The ratio of total charged values of this chain or this single step is bounded by $\phi$ (see Corollary~\ref{coro:1}).

If $q \neq e$, then no open chains exist since otherwise $e$ is a candidate packet for OPT to send. We charge OPT the value $v_q + v_f$ in this step and MG the value $v_f$. Furthermore, we split this step into two fractional steps: In one fractional step, OPT is charged a value $v_f$ and MG $v_f / \phi$. In this single fractional step, the gain ratio is $\phi$.  In another fractional step, we charge OPT the value $v_q$ and MG the value $v_f / \phi^2 \ge v_q / \phi^2$. This step maps to $q$ in MG's optimal provisional schedule at the end of this step since $e$ with $d_e \ge t$ is not the packet $q$.

\item Assume $f \neq h$.

If $q$ is not in MG's optimal provisional schedule $S$, $q$ must map to a packet $p \in S$ and $v_q \le v_e$. From Algorithm~\ref{alg:mg}, we have $v_f \ge \alpha v_e = \alpha v_q = \phi^2 v_q$. $f$ is not in any open chain (from Invariant $I_3$). We close the open chain, if any, mapping to $p$. We also charge $v_f$ to OPT in this step. The ratio of total charged values of this chain or this single step is bounded by $\phi$ (see Corollary~\ref{coro:1}).

If $q$ is in $S$, then $q$ is not in any open chain, from Invariant~\ref{remark:nodelay}. We charge OPT the value $v_q + v_f$ in this step and MG the value $v_f$. Furthermore, we split this step into two fractional steps: In one fractional step, OPT is charged a value $v_f$ and MG $v_f / \phi$. In this single fractional step, the gain ratio is $\phi$.  In another fractional step, we charge OPT the value $v_q$ and MG the value $v_f / \phi^2 \ge v_q / \phi^2$. This step maps to $q$ in MG's optimal provisional schedule at the end of this step since $e$ with $d_e \ge t$ is not the packet $q$.
\end{itemize}


\paragraph{Assume MG sends $f \neq e$ and OPT sends $q \in S$ with $d_q > d_f$}

Due to Invariants $I_2$ and $I_3$, there is no overloaded packets in MG's optimal provisional schedule. From Algorithm~\ref{alg:mg}, we have $v_q > v_f \ge \alpha v_e = \phi^2 v_e$. We start a new open chain from this step mapping to $q$ in MG's optimal provisional schedule. Note that $q$ is not an overloaded packet yet since it maps no packet in OPT's buffer.
\end{proof}


\subsection{The anti-agreeable deadline/value setting}

Consider the anti-agreeable deadline/value setting. In MG's optimal provisional schedule, for any two packets $p$ and $q$ with $d_p < d_q$, we have $v_p \ge v_q$. Applying the same proof of Theorem~\ref{theorem:antivalue}, we have
\begin{theorem}
MG is $1$-competitive for the anti-agreeable deadline/value setting when $\alpha = \infty$. MG is optimal.
\end{theorem}


\subsection{The agreeable slack-time/value setting}

\begin{lemma}
In the agreeable slack-time/value setting, if a packet $p$ is evicted out of MG's optimal provisional schedule at time $t$, then from time $t$ till $p$'s deadline $d_p$, all the MG's optimal provisional schedules do not contain any packet with a value $< v_p$.
\label{lemma:nodelaymore}
\end{lemma}

\begin{proof}
If a packet $p$ is evicted out of MG's optimal provisional schedule at time $t$, then either $d_p < t$ or in each of the buffer slots $t, \ t + 1, \ \ldots, \ d_p$, MG's current optimal provisional schedule at time $t$ buffers one packet with value $> v_p$.

In each step, MG either sends $e$ or $f \neq e$. For time $t$ when a packet $p$ is rejected, those packets unsent by MG but staying in MG's optimal provisional schedule at time $t$ are tight and cannot be shifted into later buffer slots. Note that for any two packets with the same deadline, the earlier released one has a larger slack time, hence, a larger value. Thus, the later released packet is preferred to be evicted if two packets share the same deadline and MG's optimal provisional schedule cannot accommodate both. Lemma~\ref{lemma:nodelaymore} holds.
\end{proof}

Using Lemma~\ref{lemma:nodelaymore}, we apply the proof of Theorem~\ref{theorem:newmg} directly and have
\begin{theorem}
MG is $\phi$-competitive for the agreeable slack-time/value setting when $\alpha = \beta = \phi = (1 + \sqrt{5}) / 2 \approx 1.618$.
\end{theorem}


\subsection{The anti-agreeable slack-time/value setting}

\begin{property}
Consider the anti-agreeable slack-time/value setting. In MG's optimal provisional schedule, for any two packets $p$ and $q$ with $d_p < d_q$, we have $v_p \ge v_q$.
\label{prop:slack}
\end{property}

Property~\ref{prop:slack} can be proved inductively. Assume at time $t$, Property~\ref{prop:slack} holds. Consider a packet $p$ in the optimal provisional schedule at the end of step $t$. We have $r_p \le t < d_p$. For any released packet $q$ at time $t + 1$, if $d_q < d_p$, we have $s_q = d_q - (t + 1) < d_p - t = s_p$ and $v_q > v_p$. Thus, Property~\ref{prop:slack} holds again. Property~\ref{prop:slack} results in that all the $e$-packets in the optimal provisional schedules are $\mathcal O$-packet. Applying a slightly modified version of the proof of Theorem~\ref{theorem:antivalue}, we have
\begin{theorem}
MG is $1$-competitive for the anti-agreeable slack-time/value setting when $\alpha = \infty$. MG is optimal.
\end{theorem}



\bibliographystyle{elsarticle-num}
\bibliography{complete}


\end{document}